\documentclass[journal]{IEEEtran}
\usepackage{helvet}         
\usepackage{courier}        
\usepackage{type1cm}        
\usepackage{makeidx}         
\usepackage{graphicx}        

\usepackage{multicol}        
\usepackage[bottom]{footmisc}
\usepackage{ifthen}
\usepackage[usenames,dvipsnames]{color}
\usepackage{cite}

\makeindex

\usepackage{mathtools}  
\usepackage{amssymb}  
\usepackage{amsthm}
\usepackage{amsfonts}
\usepackage{dsfont}
\usepackage{mathrsfs}
\DeclareGraphicsExtensions{.pdf,.png,.jpg,.eps}

\usepackage[caption=false,font=footnotesize]{subfig}
\usepackage{verbatim}
\usepackage{comment}
\usepackage{algorithm}
\usepackage{algorithmic}
\usepackage{multirow}
\usepackage{enumitem}
\usepackage{setspace}
\DeclareGraphicsExtensions{.pdf,.png,.jpg,.eps}

\usepackage{booktabs}
\newcommand{\rr}{{\mathbb{R}}}
\newcommand{\T}{{\text{T}}}
\newcommand{\col}{{\text{col}}}

\newboolean{showcomments}
\setboolean{showcomments}{false}
\newcommand{\fliu}[1]{\ifthenelse{\boolean{showcomments}}
	{ \textcolor{red}{(FL:  #1)}}{}}
\newcommand{\peng}[1]{\ifthenelse{\boolean{showcomments}}
	{ \textcolor{red}{(PY:  #1)}}{}}

\theoremstyle{definition}
\newtheorem{theorem}{Theorem}
\newtheorem{lemma}[theorem]{Lemma}

\theoremstyle{definition}
\newtheorem{definition}{Definition}
\newtheorem{remark}{Remark}

\newtheorem{example}{Example}
\newtheorem{assumption}{Assumption}

\newtheorem{condition}{Condition}

\begin{document}
	\setstretch{1}	
	
	\title{{Compositional} and Equilibrium-Free Stability Certification for Power Systems--Part I: Theory}
	
	\author{Peng~Yang, \IEEEmembership{Member,~IEEE}, 	
		~Xiaoyu~Peng,
		~Xi~Ru,
		~Hua~Geng, \IEEEmembership{Fellow,~IEEE},
		~Feng~Liu, \IEEEmembership{Senior~Member,~IEEE}	
				\thanks{This work was supported by the National Natural Science Foundation of China (No. 52407137) and the Young Talent Fund of Association for Science and Technology in Shaanxi, China (No. 20250431).}
				\thanks{P. Yang is with the School of Electronics and Information, Xi’an Polytechnic University, Xi'an, 710048, China. (e-mail: p-yang13@tsinghua.org.cn)}
				\thanks{X. Peng is with the Department of Electrical Engineering, Tsinghua University, Beijing, 100084, China. (e-mail: pengxy19@tsinghua.org.cn)}
			 \thanks{X. Ru is with the Department of Electrical Engineering, Tsinghua University, Beijing, 100084, China. (e-mail: rux22@mails.tsinghua.edu.cn) }
			 \thanks{H. Geng is with the Department of Automation, Tsinghua University, Beijing, 100084, China. (e-mail: genghua@tsinghua.edu.cn)} 
			 \thanks{F. Liu is with the Department of Electrical Engineering, Tsinghua University, Beijing, 100084, China. (e-mail: lfeng@tsinghua.edu.cn) (Corresponding author)}
	}
	
	\maketitle
	
	\begin{abstract}
Stability certification for large-scale power systems is challenging due to heterogeneous device dynamics and time-varying operating conditions. This two-part paper develops an equilibrium-free compositional framework based on local delta dissipativity. In Part I, we introduce the concept of \emph{local delta dissipativity (LDD)}, a localized extension of delta dissipativity in which the dissipation inequality is required only on a prescribed local state--input region, thereby relaxing the global requirement of standard delta dissipativity. Using LDD, we establish asymptotic stability with respect to an \emph{equilibrium set}, rather than with respect to a single equilibrium point, and derive {compositional} sufficient conditions that decompose into device-level LDD certificates and a network-level coupling inequality. We further specialize the results to structure-preserving power-system models described by differential-algebraic equations, yielding equilibrium-free stability certificates that are verifiable at the device level and composable through a network constraint. {Part II of this two-part paper develops the corresponding certificate computational methods and demonstrates, on modified IEEE 9-, 39-, and 118-bus systems, how the compositional theory established here enables scalable stability certification.} A single-machine single-load benchmark {illustrates} the derived conditions and the resulting certified stability region under load variations.
	\end{abstract}
	

	\IEEEpeerreviewmaketitle

	\section{Introduction}
	\label{sec:1}
\subsection{Motivation and Literature Review}
	\IEEEPARstart{T}{he} stability of modern power systems is facing unprecedented challenges driven by the rapid integration of renewable energy resources (RESs), which are predominantly interfaced via power-electronic converters. This shift leads to highly heterogeneous and nonlinear dynamics, rendering traditional stability analysis methods less effective. Conventional approaches, such as time-domain simulation\cite{Stott_Powersystemdynamic_1979}, eigenvalue analysis\cite{Sastry_Hierarchicalstabilityalert_1980}, and direct methods\cite{Chiang_Directstabilityanalysis_1995}, are inherently \textit{centralized}. They require detailed, system-wide models and intensive computation, which become increasingly impractical for large-scale grids with massive distributed energy resources due to scalability and privacy concerns \cite{liu2022stability}.
	
	To address these challenges, {compositional stability frameworks} have emerged as a promising alternative. The core idea is to infer system-wide stability from modular certificates of individual subsystems (e.g., generators, loads, converters) and their interactions. In particular, the approach based on the concept of passivity and dissipativity, proposed by J.C. Willems in the 1970s \cite{willems1972dissipative}, has been widely used to study the interconnection of nonlinear systems\cite{Bao_ProcessControlPassive_2007,vanderSchaft_L2GainPassivityTechniques_2017,arcak2016networks}. It has also been successfully applied in power systems for controller design \cite{8424071,Stegink_unifyingenergybasedapproach_2016,wang2018distributed_coping,strehle2021unified,soriano2020voltage} and stability certification \cite{Fiaz_portHamiltonianapproachpower_2013a,Caliskan_CompositionalTransientStability_2014,8890862}. Recent works have further advanced {compositional stability analysis} and passivity-based methods for microgrids and converter control in power systems \cite{phillips2024passivity,he2024passivity,malan2024passivity,fu2024distributed}. Other approaches based on Jacobian decomposition\cite{Song_DistributedFrameworkStability_2017}, eigenvalue analysis\cite{baros2020distributed}, and the small-gain and phase theorems\cite {huang2024gain} have also been proposed to decompose system-wide stability into local certificates for individual grid-connected components. 
	
	Despite these fruitful results, a critical limitation persists: most existing {compositional} stability conditions for individual devices are equilibrium-point-oriented and require explicit knowledge of the system's equilibrium point, which is less realistic in practice. For example, the local passivity conditions in \cite{8890862} and \cite{he2024passivity} require knowing the local input-state-output equilibrium; the local stability conditions in \cite{huang2024gain} consider the linearized power system model at a given equilibrium. Although focusing on a specific equilibrium facilitates stability analysis, it undermines the very essence of a {compositional} approach, since the equilibrium is a global property determined by the entire interconnected system. Consequently, it is difficult to establish universal grid codes that guarantee stability under evolving grid operating conditions. An ideal, scalable, and practical {compositional} method should be equilibrium-free at the device level, allowing local certificates to hold across a wide range of system operating conditions.

    A recently proposed concept, namely delta dissipativity \cite{schweidel2021compositional} (or Krasovskii passivity \cite{kawano2020krasovskii,kawano2023krasovskii}), provides tools for compositional stability analysis without requiring knowledge of the system-wide equilibrium. The essence of delta dissipativity is that the supply rate is defined as a function of the time derivatives of the original input and output. It requires the storage function to vanish not at a specific equilibrium, but whenever the system vector field vanishes. This enables equilibrium-independent modular properties to certify system-wide stability. However, applying this theory to large-scale power systems with complex networks and heterogeneous devices remains an open challenge. This is mainly because delta dissipativity is defined globally, requiring the dissipativity property to hold over the entire state and input space, as in other traditional definitions. For realistic power system models with strong nonlinearities and diverse components, such a global condition is often too restrictive and may fail to be satisfied, significantly limiting the applicability of delta-dissipativity-based analysis and control in practice. These limitations naturally motivate us to extend the framework from a global to a more flexible local notion, better suited to capturing the behavior of power systems within relevant operating regions.
    
    Several related concepts, namely equilibrium-independent dissipativity (EID) \cite{simpson2018equilibrium, strehle2021unified}, incremental dissipativity \cite{pavlov2008incremental}, and differential passivity \cite{6760930}, also enable equilibrium-free and compositional stability analysis. EID defines dissipativity with respect to the set of all possible equilibria, so there is no need to preselect the equilibrium permitted by the interconnection. Nevertheless, it can be difficult to characterize the set of all equilibria in complex nonlinear systems. Incremental dissipativity and differential passivity are more demanding properties that define dissipativity with respect to any two trajectories and passivity with respect to the tangent bundle of the trajectory manifold, respectively. Despite their elegance in theory, their practical scope in power systems may be limited due to their demanding requirements for subsystems. We refer interested readers to \cite{kawano2020krasovskii} for a detailed comparison among differential passivity, incremental passivity, and Krasovskii passivity.

\subsection{Contributions and Novelty}	
	In Part I of this paper, we extend the definition of delta dissipativity in \cite{schweidel2021compositional} and integrate it with our recent equilibrium-set-oriented stability theory for power systems \cite{10239448}. Specifically, we define delta dissipativity over local regions in the state--input space, which we term \textbf{local delta dissipativity (LDD)}. This extension broadens the applicability to commonly used power-system device models. Building on this tool, we derive {compositional} and equilibrium-free stability conditions for structure-preserving power systems with heterogeneous nonlinear devices.
	To explicitly position this work relative to prior results, we summarize the novelty in three aspects.
	\begin{enumerate}[label=(\roman*)]
		\item \textbf{From global to local delta dissipativity}. Existing delta-dissipativity/Krasovskii-passivity formulations are typically global in the state--input space. This paper introduces \emph{local delta dissipativity (LDD)}, in which the dissipation inequality is enforced only on a prescribed local region. This extension is technically nontrivial and is central to making equilibrium-free certificates applicable to realistic nonlinear power-system device models.
		\item \textbf{From equilibrium-point to equilibrium-set certification}. Rather than certifying stability around a pre-specified equilibrium point, the proposed analysis establishes asymptotic stability with respect to an \emph{equilibrium set}. Consequently, device-level verification does not require explicit computation of a system-wide post-fault equilibrium, which addresses a core practical limitation of equilibrium-dependent {compositional} methods.
		\item \textbf{{Compositional} conditions for structure-preserving power-system models}. Building on LDD, the paper derives sufficient system-wide conditions that decompose into device-level LDD certificates and a network-level coupling inequality. The resulting guarantee is set-wise and equilibrium-free, while remaining compatible with heterogeneous nonlinear devices and lossy, structure-preserving network models without network reduction.
	\end{enumerate}

	In summary, the key novelty is not merely the application of known dissipativity tools to power systems, but a local, equilibrium-set-oriented, and compositional stability framework that is analytically rigorous and provides the theoretical basis for the distributed computational methods developed in Part II. 	

	The proposed framework paves the way for a more scalable and adaptable approach to power system stability assessment. It also provides insights for designing device-level grid codes that can ensure system-wide stability under diverse operating scenarios. {Here, the scalability benefit arises from the compositional separation: detailed device dynamics and storage-function variables remain in the device-level LDD problems, while fixed-dimensional certificate matrices are composed through the network condition. Part II develops the corresponding computational toolkit and demonstrates its practical feasibility on modified IEEE 9-, 39-, and 118-bus structure-preserving systems containing heterogeneous nonlinear devices, static network components, and operating-condition-dependent equilibria.}
\subsection{Organization}	
	The rest of Part I is organized as follows. Section II presents the structure-preserving power-system model and formulates the research problem. Section III introduces LDD, and Section IV establishes equilibrium-set stability results for general nonlinear DAE interconnections. Section V specializes the compositional conditions to power systems and derives a region-of-attraction result. Section VI provides an illustrative example, and Section VII concludes the paper.
	
	
	\emph{Notations}: $\rr^n$ is the set of $n$-dimensional real vectors. $\text{j}=\sqrt{-1}$ denotes the imaginary unit. Let $\text{col}(x_1,x_2)=(x_1^\T,x_2^\T)^\T$ be a column vector in $\rr^{n+m}$ with $x_1\in\rr^n$ and $x_2\in\rr^m$. For a matrix $A\in\rr^{n\times n}$, $\mathrm{det}(A)$ denotes the determinant of $A$. For a symmetric matrix $A\in\rr^{n\times n}$, $A\succ(\succeq)0$ means $A$ is positive (semi-)definite, while $A\prec(\preceq)0$ means $A$ is negative (semi-) definite.
	
	\section{Problem Formulation}\label{sec:2}
	
	\subsection{Dynamic and Static Buses}\label{sec:power system model}
	Consider a structure-preserving power system that can be abstracted as an undirected graph $\cal G=(\cal V,\cal E)$, where $\cal V$ is the set of buses with $|\mathcal{V}|=N$ and $\cal E$ is the set of lines. {Following the standard usage in \cite{bergen1981structure,1085625}, ``structure-preserving'' indicates that the original buses and network topology are retained without network reduction.}
	{As is customary in electromechanical power-system stability studies, the network and its connected devices are represented by their balanced, positive-sequence, fundamental-frequency components.}

	{Each bus in $\cal G$ is associated with the positive-sequence voltage $V_{Di}+\text{j}V_{Qi}$ and current injection $I_{Di}+\text{j}I_{Qi}$ expressed in a common reference frame rotating synchronously at the nominal angular frequency. Thus, $D$ and $Q$ are the two real coordinates of the positive-sequence fundamental component.} Each bus defines either a \textit{dynamic} input-output relation between the bus voltage and the current injection (e.g., a synchronous generator) or a \textit{static} input-output relation (e.g., a constant power load). Let $\mathcal{V}_1$ denote the set of dynamic buses and $\mathcal{V}_2$ denote the set of static ones. We have $\mathcal{V}_1\cup\mathcal{V}_2=\mathcal{V}$ and $\mathcal{V}_1\cap\mathcal{V}_2=\emptyset$. {In what follows, $i\in\mathcal{V}$ denotes a generic bus, whereas $j\in\mathcal{V}_1$ and $k\in\mathcal{V}_2$ denote a dynamic bus and a static bus, respectively.}

	\begingroup
	 
	\begin{remark}[\textit{Common and local reference frames}]\label{rem:c5-reference-frame}
		The common $DQ$ frame has angle $\theta_{\rm s}(t)=\omega_{\rm s}t+\theta_0$, where $\omega_{\rm s}$ is the nominal angular frequency and the origin $\theta_0$ is an arbitrary coordinate convention. Device models may use their own local rotating $dq$ frames; their angles relative to the common frame remain device states, and \eqref{eq:c5-sg5} below gives the corresponding planar rotation for an SG. Changing $\theta_0$ by a constant $\alpha$ rotates every bus voltage and current by the same orthogonal transformation. Equivalently, $V'=e^{-\text{j}\alpha}V$ and $I'=e^{-\text{j}\alpha}I$ preserve the network equation $I=YV$ as $I'=YV'$. Therefore, the reference-angle choice changes coordinate values but not powers, equilibria under the coordinate correspondence, or stability conclusions. This common system-reference-frame formulation is standard for assembling voltage--current network models \cite{spanias2018system}.
	\end{remark}
	\endgroup
	
	{For each $j\in\mathcal{V}_1$, we consider the general nonlinear input-state-output model $H_j: u_j\mapsto y_j$ as follows.}
	
	\begin{equation}\label{eq:c5-ps-dbus}
		\left\lbrace 	
		\begin{aligned}
			\dot{x}_j&=f_j(x_j,u_j)\\
			y_j&=h_j(x_j,u_j)
		\end{aligned}\right., 
	\end{equation}
	{where ${{x}_{j}}\in {{\mathbb{R}}^{{{n}_{j}}}}$ is the state variable, ${{u}_{j}}\in {{\mathbb{R}}^{2}}$ is the input, and ${{y}_{j}}\in {{\mathbb{R}}^{2}}$ is the output. ${{f}_{j}}:\rr^{n_j}\times\rr^2\to\rr^{n_j}$ and ${{h}_{j}}:\rr^{n_j}\times\rr^2\to\rr^2$ are twice continuously differentiable functions.}
	
	{For each $k\in\mathcal{V}_2$, we consider the general nonlinear input-output model $H_k: u_k\mapsto y_k$ as follows.}
	
	\begin{equation}\label{eq:c5-ps-sbus}
		 y_k=h_k(u_k),
	\end{equation}
	{where ${{u}_{k}}\in {{\mathbb{R}}^{2}}$ is the input, ${{y}_{k}}\in {{\mathbb{R}}^{2}}$ is the output, and ${{h}_{k}}:\rr^2\to\rr^2$ is a twice continuously differentiable function.}
	
	We consider two types of input-output pairs for both dynamic and static buses. The first is defined as $${{u}_{i}}=\col{{({{V}_{Di}},{{V}_{Qi}})}},\;\;{{y}_{i}}=-\col{{({{I}_{Di}},{{I}_{Qi}})}},$$ 
    representing voltage-input/current-output behaviors. The second is defined as $${{u}_{i}}=-\col{{({{I}_{Di}},{{I}_{Qi}})}},\;\; {{y}_{i}}=\col{{({{V}_{Di}},{{V}_{Qi}})}},$$ representing current-input/voltage-output behaviors.
	
    The choice of the input-output pair in the model depends on the power device's behavior. Sometimes, it is interchangeable as shown in the following example of the 3-order synchronous generator (SG) model.

	\begin{example}\label{eg:c5-sg}
		Consider the 3-order SG model\cite{1085625}
		\begin{equation}\label{eq:c5-sg1}
			\left\lbrace 
			\begin{aligned}
				\dot{\delta}_{i}&=\omega_i \\ 
				M_i\dot{\omega}_i&=-D_i\omega_i-P_{i}^{e}+P_{i}^{m}+K_I\delta_i \\ 
				T'_{d0i}\dot{E}'_{qi}&=-E'_{qi}+I_{di}(x_{di}-x'_{di})+E_{fi}
			\end{aligned}\right.
		\end{equation}
		where $\delta_{i}$ is the rotor angle, $\omega_i$ is the frequency, $E'_{qi}$ is the $q$-axis transient voltage. $K_I\geq0$ is the coefficient of frequency integral, which aims to balance the load variation ($K_I=0$ if no integral exists). $I_{di}$, ${{I}_{qi}}$, ${{V}_{di}}$, and ${{V}_{qi}}$ are the terminal currents and voltages in the \textit{local} dq coordinate, which satisfy the following stator circuit equations
		\begin{equation}\label{eq:c5-sg2}
			\begin{bmatrix}
				V_{di}\\V_{qi}
			\end{bmatrix}=\begin{bmatrix}
				E'_{qi}\\0
			\end{bmatrix}-
			\begin{bmatrix}
				0&-x'_{di}\\x_{qi}&0
			\end{bmatrix}
			\begin{bmatrix}
				I_{di}\\I_{qi}
			\end{bmatrix},
		\end{equation}
		or equivalently,
		\begin{equation}\label{eq:c5-sg3}
			\begin{bmatrix}
				I_{di}\\I_{qi}
			\end{bmatrix}=
			\begin{bmatrix}
				0&-x'_{di}\\x_{qi}&0
			\end{bmatrix}^{-1}
			\left( \begin{bmatrix}
				V_{di}\\V_{qi}
			\end{bmatrix}-\begin{bmatrix}
				E'_{qi}\\0
			\end{bmatrix}\right) 
		\end{equation}
		The electrical active power of the generator is
		\begin{equation}\label{eq:c5-sg4}
			P_{i}^{e}=E'_{qi}I_{di}+(x'_{di}-x_{qi})I_{di}I_{qi}
		\end{equation}
		Define the transformation matrix
		\begin{equation*}
			T(\delta_i):=\begin{bmatrix}
				\cos\delta_i&\sin\delta_i\\-\sin\delta_i&\cos\delta_i
			\end{bmatrix}
		\end{equation*}
		Then, the DQ components have the following transformation relation with the local dq components
		\begin{equation}\label{eq:c5-sg5}
			\begin{bmatrix}
				I_{di}\\I_{qi}
			\end{bmatrix}=T(\delta_i)
			\begin{bmatrix}
				I_{Di}\\I_{Qi}
			\end{bmatrix},\qquad
			\begin{bmatrix}
				V_{Di}\\V_{Qi}
			\end{bmatrix}=T(\delta_i)^{-1}
			\begin{bmatrix}
				V_{di}\\V_{qi}
			\end{bmatrix},
		\end{equation}
		and
		\begin{equation}\label{eq:c5-sg6}
			\begin{bmatrix}
				I_{Di}\\I_{Qi}
			\end{bmatrix}=T(\delta_i)^{-1}
			\begin{bmatrix}
				I_{di}\\I_{qi}
			\end{bmatrix},\qquad
			\begin{bmatrix}
				V_{di}\\V_{qi}
			\end{bmatrix}=T(\delta_i)
			\begin{bmatrix}
				V_{Di}\\V_{Qi}
			\end{bmatrix}.
		\end{equation}
		
		Thus, equations \eqref{eq:c5-sg1}, \eqref{eq:c5-sg2}, \eqref{eq:c5-sg4}, and \eqref{eq:c5-sg5} constitute an input-state-output model with input $u_i=\col(-I_{Di},-I_{Qi})$, state $x_i=\col(\delta_{i},\omega_i,E'_{qi})$, and output $y_i=\col(V_{Di},V_{Qi})$. Or equivalently, eliminating $I_{di}$ and $I_{qi}$ via \eqref{eq:c5-sg3}, then \eqref{eq:c5-sg1}, \eqref{eq:c5-sg3}, \eqref{eq:c5-sg4}, and \eqref{eq:c5-sg6} constitute a model with input $u_i=\col(V_{Di},V_{Qi})$ and output $y_i=\col(-I_{Di},-I_{Qi})$.
	\end{example}
	
	\begingroup
	 
	\begin{remark}[\textit{Scope of device models}]
		The state dimension $n_j$ in \eqref{eq:c5-ps-dbus} is unrestricted; hence, increasing the order of a device model does not by itself violate the formulation. Within the balanced positive-sequence electromechanical-transient scope considered here, additional smooth device and controller dynamics may be incorporated by augmenting $x_j$, provided that the terminal behavior admits the twice continuously differentiable two-axis voltage--current representation in \eqref{eq:c5-ps-dbus} and the resulting interconnected DAE is well posed. This class includes representative SG, grid-following inverter, grid-forming inverter, and static-component models. Explicit transformations for PLL-, VSG-, conventional-droop-, and quadratic-droop-controlled inverters are given in Part II, Section II-A. Hybrid or discontinuous models and internal algebraic constraints that cannot be locally reduced to this representation are not covered directly.
	\end{remark}
	\endgroup
	
	\subsection{Feedback Interconnection via Power Network}
	All buses are interconnected through transmission lines, which are represented by the admittance matrix $Y=G+\text{j}B$, where $G\in\rr^{N\times N}$ is the conductance matrix and $B\in\rr^{N\times N}$ is the susceptance matrix. Collectively, let
    $$V_D=\col(\{V_{Di}\}_{i\in\mathcal{V}}),\;\;V_Q=\col(\{V_{Qi}\}_{i\in\mathcal{V}})$$
    $$I_D=\col(\{I_{Di}\}_{i\in\mathcal{V}}),\;\;I_Q=\col(\{I_{Qi}\}_{i\in\mathcal{V}}).$$
    The bus voltages and injected currents must satisfy Kirchhoff's current law, which gives the network equation
    $$I_D+\text{j}I_Q=Y(V_D+\text{j}V_Q),$$
    which can be written in matrix form as follows.
	\begin{equation}\label{eq:c5-net}
		\begin{bmatrix}
			I_D\\I_Q
		\end{bmatrix}=
		\underbrace{\begin{bmatrix}
				G&-B\\B&G
		\end{bmatrix}}_{:=M_Y}
		\begin{bmatrix}
			V_D\\V_Q
		\end{bmatrix}.
	\end{equation}
	
	Define the collective input of all buses as $u:=\col(\{u_i\}_{i\in\mathcal{V}})\in\rr^m$, and the collective output of all buses as $y=\col(\{y_i\}_{i\in\mathcal{V}})\in\rr^m$, where $m=2N$. Obviously, a proper permutation of $u$ and $y$ yields $-I_D,-I_Q,V_D,V_Q$. This defines a linear relation between $-I_D,-I_Q,V_D,V_Q$ and $u,y$, which is denoted as
	\begin{equation}\label{eq:c5-ABuy}
		\begin{bmatrix}
			-I_D\\-I_Q
		\end{bmatrix}=A_Iu+B_Iy,\quad 
		\begin{bmatrix}
			V_D\\V_Q
		\end{bmatrix}=A_Vu+B_Vy,
	\end{equation}
	\begingroup
	 
	To make their structure explicit, let $\Pi_{\mathrm{DQ}}\in\{0,1\}^{m\times m}$ denote the permutation matrix that converts the bus-wise ordering into the D/Q-component-wise ordering, i.e.,
	\begin{equation*}
		\begin{aligned}
			&\Pi_{\mathrm{DQ}}\col(z_{D1},z_{Q1},\ldots,z_{DN},z_{QN})\\
			&\qquad=\col(z_{D1},\ldots,z_{DN},z_{Q1},\ldots,z_{QN}).
		\end{aligned}
	\end{equation*}
	For each $i\in\mathcal{V}$, let $\sigma_i=0$ if bus $i$ adopts the voltage-input/current-output representation and $\sigma_i=1$ if it adopts the current-input/voltage-output representation, and define
	\begin{equation*}
		D_\sigma:=\operatorname{blkdiag}(\sigma_1I_2,\ldots,\sigma_NI_2).
	\end{equation*}
	Then the four matrices in \eqref{eq:c5-ABuy} are
	\begin{equation*}
		\begin{aligned}
			A_I&=\Pi_{\mathrm{DQ}}D_\sigma, &
			B_I&=\Pi_{\mathrm{DQ}}(I_m-D_\sigma),\\
			A_V&=\Pi_{\mathrm{DQ}}(I_m-D_\sigma), &
			B_V&=\Pi_{\mathrm{DQ}}D_\sigma.
		\end{aligned}
	\end{equation*}
	Thus, $A_I,B_I,A_V,$ and $B_V$ are binary selection matrices, with $A_I=B_V$ and $B_I=A_V$. Although they are not individually permutation matrices in general, the combined matrix $\left[\begin{smallmatrix}A_I&B_I\\A_V&B_V\end{smallmatrix}\right]$ is a permutation matrix mapping $\col(u,y)$ to $\col(-I_D,-I_Q,V_D,V_Q)$.
	\endgroup
	
	Define the network coupling $H_{net}:u_{net}\mapsto y_{net}$ with 
	\begin{equation}\label{eq:net-inout}
		u_{net}=y,\quad y_{net}=-u.
	\end{equation}
	Substituting \eqref{eq:c5-ABuy} and \eqref{eq:net-inout} into \eqref{eq:c5-net} yields
	\begin{equation*}
		A_Iy_{net}-B_Iu_{net}=M_{Y}\left(-A_Vy_{net}+B_Vu_{net}\right).
	\end{equation*}
	\begingroup
	 
	\begin{assumption}[\textit{Network input--output well-posedness}]\label{as:c5-network-wellposed}
		The matrix $A_I+M_YA_V$ is nonsingular, i.e.,
		\begin{equation*}
			\operatorname{rank}(A_I+M_YA_V)=m.
		\end{equation*}
	\end{assumption}
	\begin{remark}
		Assumption \ref{as:c5-network-wellposed} ensures that, for every $u_{net}$, the network equation uniquely determines $y_{net}$. If $A_I+M_YA_V$ is singular, the linear equation above is either inconsistent or, whenever it is consistent, has infinitely many solutions; hence, $h_{net}$ cannot be defined as a single-valued mapping. From a circuit perspective, such singularity corresponds to an incomplete or dependent specification of the network port variables. It may arise, for example, from an unreferenced floating island or a port selection that leaves a network voltage mode unconstrained. These degenerate configurations are not expected to arise in a normally connected and properly referenced power network with standard physical port selections. Since $A_I$, $A_V$, and $M_Y$ are constant, Assumption \ref{as:c5-network-wellposed} can be checked offline and only needs to be revalidated when the network topology or the port orientations change.
	\end{remark}
	\endgroup
	{Under Assumption \ref{as:c5-network-wellposed}, the static input--output equation of the network coupling is}
	\begin{equation}\label{eq:c5-ps-HN}
		h_{net}(u_{net}):=\left(A_I+M_YA_V\right)^{-1}\left(B_I+M_YB_V\right)u_{net}.
	\end{equation}
	We define
	\begin{equation}\label{eq:c5-partialhnet}
		C:=\frac{\partial h_{net}}{\partial u_{net}}=\left(A_I+M_YA_V\right)^{-1}\left(B_I+M_YB_V\right).
	\end{equation}
	All bus subsystems \eqref{eq:c5-ps-dbus} \eqref{eq:c5-ps-sbus} are interconnected with the network coupling \eqref{eq:c5-ps-HN} in a feedback manner with $u_{net}=y$ and $u=-y_{net}$, as shown in Fig. \ref{fig:c5-feedback}.
	\begin{figure}[htb]
		\centering
		\includegraphics[width=0.6\hsize]{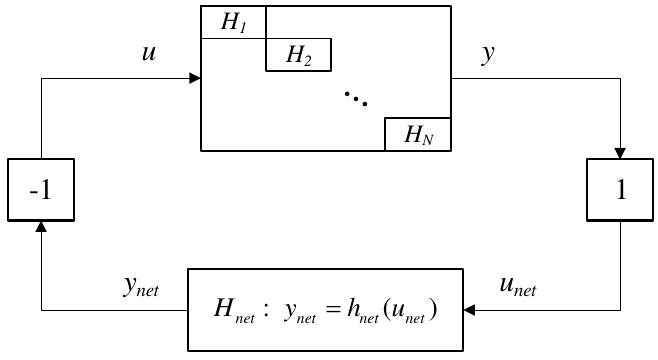}
		\caption{Feedback interconnection between bus subsystems and the network.}
		\label{fig:c5-feedback}
	\end{figure}
	
	\begingroup Define the collective state variable $x:=\col(\{x_j\}_{j\in\mathcal{V}_1})\in\rr^n$, where $n=\sum_{j\in\mathcal{V}_1}n_j$. Define functions $$f(x,u):=\col(\{f_j\}_{j\in\mathcal{V}_1}),\;\;h(x,u):=\col(\{h_i\}_{i\in\mathcal{V}}),$$\endgroup
    and
    $$g(x,u):=u+h_{net}(h(x,u)).$$
    Then the overall power-system model can be written compactly as the following differential-algebraic equations (DAEs).
	\begin{equation}\label{eq:c5-system}
		\left\lbrace
		\begin{aligned}
			\dot{x}=&f(x,u)\\
			0=&g(x,u)
		\end{aligned}\right..
	\end{equation}
	This gives the standard structure-preserving DAE model of power systems \cite{bergen1981structure}, where $x\in\rr^n$ is the state variable, $u\in\rr^m$ is the algebraic variable, $f$ and $g$ are continuously differentiable functions.
	
	\begin{remark}[\textit{Selection of coordinates}]
		The most commonly used model to characterize power grid coupling is the power flow equations that map the bus voltage magnitude and angle $(V,\theta)$ to the power injection $(P,Q)$, see \cite{8890862} for example. Despite its widespread use, it yields a nonlinear coupling relation that complicates {compositional} stability analysis. To get a linear network equation, we follow \cite{spanias2018system} and model the power grid by voltage-current equations in the common DQ framework. This enables equilibrium-free, {compositional} conditions to certify power system stability, as elaborated in Sections \ref{sec:stability conditions for power system} and \ref{sec:transient stability}.
	\end{remark}
	
	\subsection{Problem Statement}
	{In this paper, we aim to develop a \emph{compositional} and \emph{equilibrium-free} framework for certifying the asymptotic stability of \eqref{eq:c5-system}. By \emph{compositional}, we mean that a system-wide stability certificate is constructed from: i) locally verifiable certificates for the bus subsystems \eqref{eq:c5-ps-dbus} and \eqref{eq:c5-ps-sbus}; and ii) a network-level coupling condition associated with \eqref{eq:c5-ps-HN}. By \emph{equilibrium-free}, we mean that the device-level LDD certificates and the network coupling condition can be formulated without preselecting a particular system-wide equilibrium. This distinction is important in a compositional setting because the steady-state input and state of an individual device are jointly determined by the complete network interconnection and therefore cannot, in general, be obtained from the isolated device model. The term does not mean that equilibria are absent from the system or from the stability conclusion; rather, the resulting guarantee concerns the equilibrium set contained in the certified region. Consequently, we develop an equilibrium-set-oriented stability theory by characterizing which set of equilibria is asymptotically stable\footnote{The asymptotic stability of the set $\mathscr{E}$ means \cite{khalil2002nonlinear} that for any $\varepsilon>0$, there exists $\delta>0$ such that $\text{dist}((x(0),u(0)),\mathscr{E})<\delta$ implies $\text{dist}\left((x(t),u(t)),\mathscr{E}\right)<\varepsilon$, $\forall t\geq0$, and that as $t\to\infty$, $\text{dist}\left((x(t),u(t)), \mathscr{E}\right)\to0$. Here, $\text{dist}(p,A)=\inf_{q\in A}\|p-q\|$ denotes the distance from the point $p$ to the set $A$.} under the established compositional stability conditions.}
	
	\section{Local Delta Dissipativity}\label{sec:5.2}
	This section defines the concept of \textbf{local delta dissipativity} for general nonlinear input--state--output dynamic systems and input--output static systems. Building on this foundation, we derive {compositional} stability conditions for interconnected systems described by DAEs. These results serve as fundamental mathematical tools for applications in power systems. To present the theoretical results more effectively, we use generalized system models in this section. Consequently, the conclusions are not only applicable to the power-system models in Section \ref{sec:2} but also extend to more general nonlinear system models.
	\subsection{Definitions of Local Delta Dissipativity}
	Consider a nonlinear input-state-output dynamic system as follows.
	\begin{equation}\label{eq:dbus-noi}
		\left\lbrace 	
		\begin{aligned}
			\dot{x}&=f(x,u)\\
			y&=h(x,u),
		\end{aligned}\right. 
	\end{equation}
	where $x\in\rr^n$, $u\in\rr^m$, $y\in\rr^m$, $f:\rr^n\times\rr^m\to\rr^n$, and $h:\rr^n\times\rr^m\to\rr^m$ are continuously differentiable functions.
	\begin{definition}[\textit{LDD for dynamic systems}]
    \label{de:IODP}
		The system \eqref{eq:dbus-noi} is locally delta dissipative with supply rate $w:\rr^m\times\rr^m\to\rr$ on $\mathcal{D}\subset\rr^n\times\rr^m$ if there exists a continuously differentiable storage function $S:\rr^n\times\rr^m\to\rr$ such that, for some class $\mathcal{K}$ functions $\alpha$, $\beta$, and $\gamma$, it holds that:
		\begin{enumerate}
			\item For all $(x,u)\in\mathcal{D}$,
			\begin{equation*}
				\alpha\left(\|f(x,u)\|\right)\leq S(x,u)\leq \beta\left(\|f(x,u)\|\right)
			\end{equation*}
			\item For all $(x,u)\in\mathcal{D}$ and all $\dot{u}\in\rr^m$,
			\begin{equation}\label{eq:c5-dotS}
				\begin{aligned}
					\frac{\partial S(x,u)}{\partial x}f(x,u)+\frac{\partial S(x,u)}{\partial u}\dot{u}
					\leq	w(\dot{u},\dot{y})
					-\gamma\left(\|f(x,u)\|\right)
				\end{aligned}
			\end{equation}
			where $\dot{y}:=\frac{\partial h(x,u)}{\partial x}f(x,u)+\frac{\partial h(x,u)}{\partial u}\dot{u}$
		\end{enumerate}
		In particular, the system is said to be globally delta dissipative if $\mathcal{D}=\rr^n\times\rr^m$.
	\end{definition}
	
	The delta dissipativity can be viewed as a special type of dissipativity with differential input and differential output, as shown in Figure \ref{fig:c5-IODP}. Noting that $\dot{x}=f(x,u)$ and $\dot{S}(x,u)=\frac{\partial S(x,u)}{\partial x}\dot{x}+\frac{\partial S(x,u)}{\partial u}\dot{u}$, the dissipation inequality \eqref{eq:c5-dotS} is equivalent to
	\begin{equation*}
		\dot{S}(x,u)\leq
		w(\dot{u},\dot{y})-\gamma\left(\|\dot{x}\|\right)
	\end{equation*}
	That is, the dissipation is only related to the differentials, including the supply rate consisting of $\dot{u}$ and $\dot{y}$, and the class $\mathcal{K}$ function of $\dot{x}$. Such a definition of dissipativity is independent of the steady-state values of $u,y,x$, which will help establish equilibrium-independent stability conditions.
	\begin{figure}[htb]
		\centering
		\includegraphics[width=0.9\hsize]{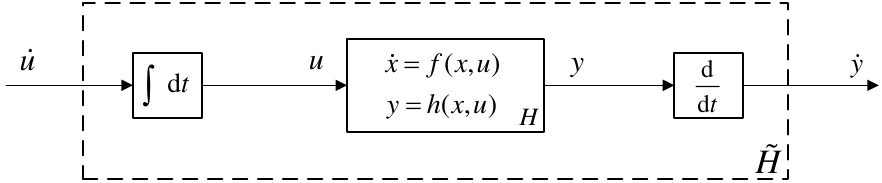}
		\caption{The LDD of $H$ is equivalent to the classical dissipativity of $\tilde{H}$.}
		\label{fig:c5-IODP}
	\end{figure}
	
	In this paper, we restrict our attention to quadratic supply rates of the form
	\begin{equation}\label{eq:X}
		w(\dot{u},\dot{y})=
		\begin{bmatrix}
			\dot{u}\\\dot{y}
		\end{bmatrix}^\text{T}
		X
		\begin{bmatrix}
			\dot{u}\\\dot{y}
		\end{bmatrix},
	\end{equation}
	where $X\in\rr^{2m\times 2m}$ is a symmetric matrix. If the system is locally delta dissipative with the quadratic supply rate \eqref{eq:X} on $\mathcal{D}$, we denote it as LDD($X,\mathcal{D}$) for short. The quadratic form contains some common types of dissipation as special cases, e.g., passivity, passivity with indices, and $L_2$-gain. See \cite{simpson2018equilibrium} for further discussion.
	
	Now consider a static input-output system
	\begin{equation}\label{eq:sbus-noi}
		y=h(u),
	\end{equation}
	where $u\in\rr^m$, $y\in\rr^m$, and $h:\rr^m\to\rr^m$ is a continuously differentiable function. We define local delta dissipativity of \eqref{eq:sbus-noi} by taking the storage function as zero in Definition \ref{de:IODP}.
	\begin{definition}[\textit{LDD for static systems}]
    \label{de:IODPs}
		The system \eqref{eq:sbus-noi} is locally delta dissipative with supply rate $w:\rr^m\times\rr^m\to\rr$ on $\mathcal{D}\subset\rr^m$ if
		\begin{equation*}
			w(\dot{u},\dot{y})\geq0,\;\;\forall u\in\mathcal{D},\;\forall \dot{u}\in\rr^m
		\end{equation*}
		where $\dot{y}:=\frac{\partial h(u)}{\partial u}\dot{u}$. In particular, the system is said to be globally delta dissipative if $\mathcal{D}=\rr^m$.
	\end{definition}
	
	For a quadratic supply rate of the form \eqref{eq:X}, the condition becomes
	\begin{equation}\label{eq:c5-Ps}
		\begin{bmatrix}
			I\\\frac{\partial h(u)}{\partial u}
		\end{bmatrix}^\text{T}
		X
		\begin{bmatrix}
			I\\\frac{\partial h(u)}{\partial u}
		\end{bmatrix}\succeq0,\;\;\forall u\in\mathcal{D},
	\end{equation}
	since
	\begin{equation*}
		\begin{bmatrix}
			\dot{u}\\\dot{y}
		\end{bmatrix}^\text{T}
		X
		\begin{bmatrix}
			\dot{u}\\\dot{y}
		\end{bmatrix}=
		\dot{u}^\text{T}\begin{bmatrix}
			I\\\frac{\partial h(u)}{\partial u}
		\end{bmatrix}^\text{T}
		X
		\begin{bmatrix}
			I\\\frac{\partial h(u)}{\partial u}
		\end{bmatrix}\dot{u}.
	\end{equation*}
	Again, we denote a static system that is locally delta dissipative with the quadratic supply rate \eqref{eq:X} on $\mathcal{D}$ as LDD($X,\mathcal{D}$).
	
	Note that the choice of $X$ for the same system is often not unique. In particular, for any $X_1\prec X_2$, LDD($X_1,\mathcal{D}$) implies LDD($X_2,\mathcal{D}$). Hence, one should choose the most ``negative'' $X$ to accurately quantify the dissipativity of a system.
	\begin{remark}[\textit{Comparison of related definitions}]
		Dissipativity with both differentiated inputs and outputs, to the best of our knowledge, first appeared in the study of evolutionary games to analyze convergence\cite{park2018passivity,fox2013population,arcak2020dissipativity}. It was later formally introduced as delta dissipativity in \cite{schweidel2021compositional} for general nonlinear systems. Our definitions are a variant of \cite{schweidel2021compositional} and differ in two aspects. First, three class-$\mathcal{K}$ functions $\alpha$, $\beta$, and $\gamma$ are introduced to bound the storage function $S(x,u)$ and the dissipation term, which enables region-of-attraction analysis via \cite[Theorem 6]{10239448}. Second, Conditions 1) and 2) in Definitions \ref{de:IODP} and \ref{de:IODPs} are required to hold only on a region $\mathcal{D}\subset\rr^n\times\rr^m$. In contrast, \cite{schweidel2021compositional} requires similar conditions to hold for all $(x,u)\in\rr^n\times\rr^m$, which corresponds to the special case $\mathcal{D}=\rr^n\times\rr^m$. Introducing $\mathcal{D}$ converts global requirements into local ones, thereby making delta dissipativity less demanding and facilitating its application to power systems. Indeed, the case studies in Section \ref{sec:case} will show that the synchronous generator model only permits LDD.
	\end{remark}
	
	\subsection{Feedback Interconnection of LDD Subsystems}
	Consider the feedback interconnection of $N=N_d+N_s$ subsystems satisfying LDD($X_i,\mathcal{D}_i$). Without loss of generality, we assume that the first $N_d$ subsystems are dynamical, denoted by $H_i$ for $i=1,\dots,N_d$, and given by
	\begin{equation}\label{eq:c5-dbus}
		\left\lbrace 
		\begin{aligned}
			\dot{x}_i&=f_i(x_i,u_i)\\
			y_i&=h_i(x_i,u_i)
		\end{aligned}\right.,\qquad i=1,\dots,N_d
	\end{equation}
	where $x_i\in \rr^{n_i}$, $u_i\in \rr^{m_i}$, and $y_i\in \rr^{m_i}$. 
	The remaining $N_s$ subsystems are static, denoted by $H_i$ for $i=N_d+1,\dots,N$, and given by
	\begin{equation}\label{eq:c5-sbus}
		y_i=h_i(u_i),\qquad i=N_d+1,\dots,N
	\end{equation}
	where $u_i\in \rr^{m_i}$ and $y_i\in \rr^{m_i}$. Let $n=n_1+\dots+n_{N_d}$, $m=m_1+\dots+m_{N}$ and define the collective state, input, and output variables as $x:=\col(x_1,\dots,x_{N_d})\in\rr^n$, $u:=\col(u_1,\dots,u_{N})\in\rr^m$, and $y:=\col(y_1,\dots,y_{N})\in\rr^m$. Define the function
	\begin{equation*}
		f(x,u):=\col\left(f_1(x_1,u_1),\dots,f_{N_d}(x_{N_d},u_{N_d})\right)
	\end{equation*}
	and
	\begin{equation*}
		\begin{aligned}
			h(x,u):=\col(&h_1(x_1,u_1),\dots,h_{N_d}(x_{N_d},u_{N_d}),\\
			&h_{N_d+1}(u_{N_d+1}),\dots,h_N(u_{N}))
		\end{aligned}
	\end{equation*}
	
	As a generalization of the power network \eqref{eq:c5-ps-HN}, we consider a nonlinear network coupling that characterizes the interaction of the above $N$ subsystems, which reads $H_{net}:u_{net}\mapsto y_{net}$
	\begin{equation}\label{eq:c5-HN}
		y_{net}=h_{net}(u_{net})
	\end{equation}
	where $u_{net}\in\rr^m$, $y_{net}\in\rr^m$, and $h_{net}$ is continuously differentiable. The subsystems \eqref{eq:c5-dbus} and \eqref{eq:c5-sbus} are interconnected with \eqref{eq:c5-HN} in feedback with $u_{net}=y$ and $u=-y_{net}$, as shown in Fig. \ref{fig:c5-feedback}.

	Define the function
	\begin{equation}\label{eq:g}
		g(x,u):=h_{net}(h(x,u))+u.
	\end{equation}
	Then the above feedback-interconnected system can be written compactly as the following DAEs:
	\begin{equation}\label{eq:c5-dae}
		\left\lbrace 
		\begin{aligned}
			\dot{x}&=f(x,u)\\
			0&=g(x,u)
		\end{aligned}\right.
	\end{equation}
	where $x$ is the state variable and $u$ is the algebraic variable of the DAEs. Denote the algebraic manifold by
	\begin{equation*}
		\mathscr{G}:=\left\{(x,u)\in\rr^n\times\rr^m|0=g(x,u)\right\}
	\end{equation*}
	Let $\mathcal{D}:=\mathcal{D}_1\times\dots\times\mathcal{D}_N$ denote the Cartesian product of each region in which the delta dissipativity holds for each subsystem. Clearly, we have $\mathcal{D}\subset\rr^n\times\rr^m$. Let $\mathcal{D}_G:=\mathcal{D}\cap\mathscr{G}$ denote the intersection of $\mathcal{D}$ and the algebraic manifold.
	{The sets $\mathcal{D}_i$ and $\mathcal{D}$ specify the validity domains of the local LDD conditions and are not assumed to be invariant. Moreover, $u_i$ represents a physical port variable determined by the interconnection, rather than an additional control variable.}
	\begin{equation}\label{eq:c5-Dy}
		\mathcal{D}_y:=\left\{y\in\rr^m|y=h(x,u), (x,u)\in\mathcal{D}_G\right\}
	\end{equation}
	denote the image of $\mathcal{D}_G$ under the map $h(x,u)$. 
	
	Define the equilibrium set in the dissipative region as
	\begin{equation}
    \label{eq:E}
		\mathscr{E}:=\left\{(x,u)\in\mathcal{D}_G|f(x,u)=0\right\}
	\end{equation}
    We assume that the interconnected system \eqref{eq:c5-dae} possesses at least one equilibrium in $\mathcal{D}_G$ and the set of equilibria is bounded. 
	\begin{assumption}[\textit{Equilibrium}]
    \label{as:c5-eqexist}
		$\mathscr{E}$ is non-empty and bounded.
	\end{assumption}
	
	We also assume as standard that in the closure of $\mathcal{D}$, the following non-singular condition holds, such that the interconnection is well-posed.
	
	\begin{assumption}[\textit{Non-singularity}]
    \label{as:c5-nonsingular}
		For any $(x,u)\in\overline{\mathcal{D}}$, $\det\left(\frac{\partial g(x,u)}{\partial u}\right)\neq0$. {Here, $\overline{\mathcal{D}}$ denotes the closure of $\mathcal{D}$.} 
	\end{assumption}

	{Throughout the continuous-time analysis, we consider classical solutions of a fixed DAE model for which the input and output signals are continuously differentiable. Thus, $\dot{u}$ and $\dot{y}$ in the LDD supply rates are ordinary time derivatives.}

	\begin{remark}
	Assumption \ref{as:c5-nonsingular} is commonly used in studies on the power systems transient stability problem \cite{6750746,davy1997lyapunov,dib2009globally,wan2017nonlinear,stegink2017stabilization}. It guarantees the existence and uniqueness of the solution to \eqref{eq:c5-dae} given any initial value $(x_0,u_0)\in \mathcal{D}_G$ (see \cite{10239448} and \cite{hill1990stability} for details). Note that we do not require $\mathcal{D}_G$ being a positively invariant set, i.e., we do not require any solution starting in $\mathcal{D}_G$ should stay in $\mathcal{D}_G$ for all future time. But we do restrict the permissible equilibria in $\mathcal{D}_G$. Solutions may leave $\mathcal{D}_G$ at the so-called impasse surface defined by $\mathrm{det}(\frac{\partial g}{\partial u})=0$, at which the singularity-induced bifurcation takes place \cite{marszalek2005singularity}. Solutions approaching the impasse surface are typically associated with the short-term voltage collapse in power systems \cite{hiskens1989energy,venkatasubramanian1995local}, which is beyond the scope of this paper.
\end{remark}
	\section{Stability Induced by Local Delta Dissipativity}
	In this section, we prove that the local delta dissipativity of each subsystem \eqref{eq:c5-dbus} \eqref{eq:c5-sbus}, together with a certain property of the network coupling \eqref{eq:c5-HN}, can yield the stability of the interconnected system \eqref{eq:c5-dae}. 
    To this end, we first show in the following Lemma that the summation of local storage functions, as well as their derivatives, can be bounded by a certain class $\mathcal{K}$ functions.
	
	\begin{lemma} [\textit{Feedback interconnection of LDD subsystems}]
    \label{th:c5-couple}
		Consider the interconnected DAEs system \eqref{eq:c5-dae}. Assume that each dynamic subsystem \eqref{eq:c5-dbus} satisfies LDD($X_i,\mathcal{D}_i$), with the corresponding storage function $S_i(x_i,u_i)$, each static subsystem \eqref{eq:c5-sbus} satisfies LDD($X_i,\mathcal{D}_i$). Consider $\mathcal{D}_y$ defined as in equations \eqref{eq:c5-Dy}. If there exist $p_i>0$  $i=1,\dots,N$ such that for any $y\in\mathcal{D}_y$
		\begin{equation}\label{eq:c5-global}
			\begin{bmatrix}
				-\frac{\partial h_{net}}{\partial y}\\I
			\end{bmatrix}^\text{T}P_\pi^\text{T}\text{blkdiag}(p_1X_1,\dots,p_{N}X_{N})P_\pi\begin{bmatrix}
				-\frac{\partial h_{net}}{\partial y}\\I
			\end{bmatrix}
			\preceq0,
		\end{equation}
		where $P_\pi$ is the permutation matrix such that $$P_\pi\col(u,y)=\col(u_1,y_1\dots,u_N,y_N).$$ 
        Then for the function $S(x,u):=\sum_{i=1}^{N_d}p_iS_i(x_i,u_i)$, there exist three class $\mathcal{K}$ functions $\alpha$, $\beta$, and $\gamma$ such that
		\begin{equation*}
			\alpha(\|f(x,u)\|)\leq S(x,u)\leq\beta(\|f(x,u)\|),\forall (x,u)\in\mathcal{D}_G 
		\end{equation*}
		\begin{equation*}
			\dot{S}(x,u)\leq-\gamma(\|f(x,u)\|),\;\forall (x,u)\in\mathcal{D}_G 
		\end{equation*}
	\end{lemma}
	\begin{proof}
		1) By Definition \ref{de:IODP}, there exist class $\mathcal{K}$ functions $\alpha_i$ and  $\beta_i$, $i=1,\dots,N_d$ such that for any $(x,u)\in\mathcal{D}_G$ we have
		\begin{equation*}
			\sum_{i=1}^{N_d}p_i\alpha_i(\|f_i\|)\leq S(x,u)\leq\sum_{i=1}^{N_d}p_i\beta_i(\|f_i\|)
		\end{equation*}
		For $r\in[0,\infty)$, define the set of $N_d$-dimensional positive quadrant with radius $r$ as $$D(r):=\left\{(r_1,\dots,r_{N_d})\in[0,\infty)^{N_d}\big|\|(r_1,\dots,r_{N_d})\|=r\right\}$$
		Define the functions
		\begin{equation*}
			\alpha(r):=\min_{(r_1,\dots,r_{N_d})\in D(r)}\left(p_1\alpha_1(r_1)+\dots+p_{N_d}\alpha_{N_d}(r_{N_d})\right)
		\end{equation*}
		and
		\begin{equation*}
			\beta(r):=p_1\beta_1(r)+\dots+p_{N_d}\beta_{N_d}(r)
		\end{equation*}
		Since $p_i>0$, one can verify that $\alpha(r)$ and $\beta(r)$ are both class $\mathcal{K}$ functions on $[0,\infty)$. Noting $f=\col(f_1,\dots,f_{N_d})$, we have
		\begin{equation*}
			\alpha(\|f\|)\leq\sum_{i=1}^{N_d}p_i\alpha_i(\|f_i\|)
		\end{equation*}
		\begin{equation*}
			\sum_{i=1}^{N_d}p_i\beta_i(\|f_i\|)\leq\sum_{i=1}^{N_d}p_i\beta_i(\|f\|)=\beta(\|f\|)
		\end{equation*}
		Hence, for any $(x,u)\in\mathcal{D}_G$, we have $\alpha(\|f(x,u)\|)\leq S(x,u)\leq\beta(\|f(x,u)\|)$.
		
		2) By Definitions \ref{de:IODP} and \ref{de:IODPs}, there exists class $\mathcal{K}$ functions $\gamma_i$, $i=1,\dots,N_d$ such that for any $(x,u)\in\mathcal{D}_G$ we have
		\begin{equation*}
			\dot{S}(x,u)\leq
			\sum_{i=1}^{N}p_i\begin{bmatrix}
				\dot{u}_i\\\dot{y}_i
			\end{bmatrix}^\text{T}
			X_i
			\begin{bmatrix}
				\dot{u}_i\\\dot{y}_i
			\end{bmatrix}-\sum_{i=1}^{N_d}p_i\gamma_i\left(\|f_i\|\right)
		\end{equation*}
		It follows from the definition of $P_\pi$ that
		\begin{equation*}
			\begin{aligned}
				\sum_{i=1}^{N}p_i\begin{bmatrix}
					\dot{u}_i\\\dot{y}_i
				\end{bmatrix}^\text{T}
				&X_i
				\begin{bmatrix}
					\dot{u}_i\\\dot{y}_i
				\end{bmatrix}\\
				&=\begin{bmatrix}
					\dot{u}\\ \dot{y}
				\end{bmatrix}^\text{T}P_\pi^\text{T}\text{blkdiag}(p_1X_1,\dots,p_{N}X_{N})P_\pi\begin{bmatrix}
					\dot{u}\\ \dot{y}
				\end{bmatrix}
			\end{aligned}
		\end{equation*}
		On the other hand, the network coupling \eqref{eq:c5-HN} determines the algebraic relation $-u=h_{net}(y)$ between $u$ and $y$. Hence,
		\begin{equation*}
			\begin{bmatrix}
				\dot{u}\\ \dot{y}
			\end{bmatrix}=\begin{bmatrix}
				-\frac{\partial h_{net}(y)}{\partial y}\\I
			\end{bmatrix}\dot{y}
		\end{equation*}
		Thus, by \eqref{eq:c5-global} and noting that $\mathcal{D}_y$ is the image of $\mathcal{D}_G$ under $h(x,u)$, it follows that for an arbitrary $(x,u)\in\mathcal{D}_G$ we have
		\begin{equation*}
			\sum_{i=1}^{N}p_i\begin{bmatrix}
				\dot{u}_i\\\dot{y}_i
			\end{bmatrix}^\text{T}
			X_i
			\begin{bmatrix}
				\dot{u}_i\\\dot{y}_i
			\end{bmatrix}\leq0
		\end{equation*}
		
		Similarly, for any $r\in[0,\infty)$ define the function
		\begin{equation*}
			\gamma(r):=\min_{(r_1,\dots,r_{N_d})\in D(r)}\left(p_1\gamma_1(r_1)+\dots+p_{N_d}\gamma_{N_d}(r_{N_d})\right)
		\end{equation*}
		We can verify that $\gamma(r)$ is a class $\mathcal{K}$ function on $[0,\infty)$ and
		\begin{equation*}
			\gamma(\|f\|)\leq\sum_{i=1}^{N_d}p_i\gamma_i(\|f_i\|)
		\end{equation*}
		Hence, for any $(x,u)\in\mathcal{D}_G$ we have $\dot{S}(x,u)\leq-\gamma(\|f(x,u)\|)$, which completes the proof.
	\end{proof}
    
    \subsection{Stability of Interconnected LDD Subsystems}
	Now, we are ready to present the main theorem of this paper.
	\begin{theorem}[\textit{Stability}]
    \label{th:c5-stability}
		Consider the system \eqref{eq:c5-dae} that satisfies Assumptions \ref{as:c5-network-wellposed}-\ref{as:c5-nonsingular}. If 
		\begin{enumerate}
			\item Each dynamic subsystem \eqref{eq:c5-dbus} is LDD($X_i,\mathcal{D}_i$), with the corresponding storage function $S_i(x_i,u_i)$;
			\item Each static subsystem \eqref{eq:c5-sbus} is LDD($X_i,\mathcal{D}_i$);
			\item There exist $p_i>0$  $i=1,\dots,N$ such that for any $y\in\mathcal{D}_y$
			\begin{equation*}
				\begin{bmatrix}
					-\frac{\partial h_{net}}{\partial y}\\I
				\end{bmatrix}^\text{T}P_\pi^\text{T}\text{blkdiag}(p_1X_1,\dots,p_{N}X_{N})P_\pi\begin{bmatrix}
					-\frac{\partial h_{net}}{\partial y}\\I
				\end{bmatrix}
				\preceq0,
			\end{equation*}
			where $P_\pi$ is the permutation matrix such that $P_\pi\col(u,y)=\col(u_1,y_1,\dots,u_N,y_N)$.
		\end{enumerate}
		Then the following conclusions hold:
		\begin{enumerate}
			\item The equilibrium set $\mathscr{E}$ is asymptotically stable;
			\item Every isolated equilibrium point in $\mathscr{E}$ is asymptotically stable;
			\item Define the weighted sum of storage functions $$S(x,u):=\sum_{i=1}^{N_d}p_iS_i(x_i,u_i)$$
            and let $$\bar{l}:=\min_{(x,u)\in\partial\mathcal{D}_G}S(x,u).$$
            For any level value $l<\bar{l}$, the sub-level set ${S}_l^{-1}$ is a positive invariant set of the system \eqref{eq:c5-dae}, and any initial point starting in ${S}_l^{-1}$ converges to $\mathscr{E}$ as $t\to+\infty$.
		\end{enumerate}
	\end{theorem}
	\begin{proof}
		1) For any $\varepsilon>0$, choose $r\in(0,\varepsilon]$ such that
		\begin{equation*}
			B_r:=\left\{(x,u)\in\rr^n\times\rr^m|\text{dist}\left((x,u),\mathscr{E}\right)\leq r\right\}\subset \mathcal{D}_G
		\end{equation*}
		Let $a:=\min_{\text{dist}\left((x,u),\mathscr{E}\right)=r}S(x,u)$.
		It follows from Lemma \ref{th:c5-couple} that $a>0$. For any $b\in(0,a)$, consider the sub-level set\footnote{The sub-level set $S_b^{-1}$ may have multiple disconnected branches in $\rr^n\times\rr^m$. Here, only those branches that intersect with $\mathcal{D}_G$ are considered. And by Assumption \ref{as:c5-eqexist} such a branch must exist.}
		\begin{equation*}
			S_b^{-1}=\left\{(x,u)\in\rr^n\times\rr^m|S(x,u)\leq b\right\}
		\end{equation*}
		Clearly, we have $\mathscr{E}\subset S_b^{-1}$. Moreover, $S_b^{-1}$ must lie in the interior of $B_r$, i.e., $S_b^{-1}\subset \text{Int}(B_r)$, as shown in Fig. \ref{fig:C5-relation of sets}.
		\begin{figure}[htb]
			\centering			\includegraphics[width=0.5\linewidth]{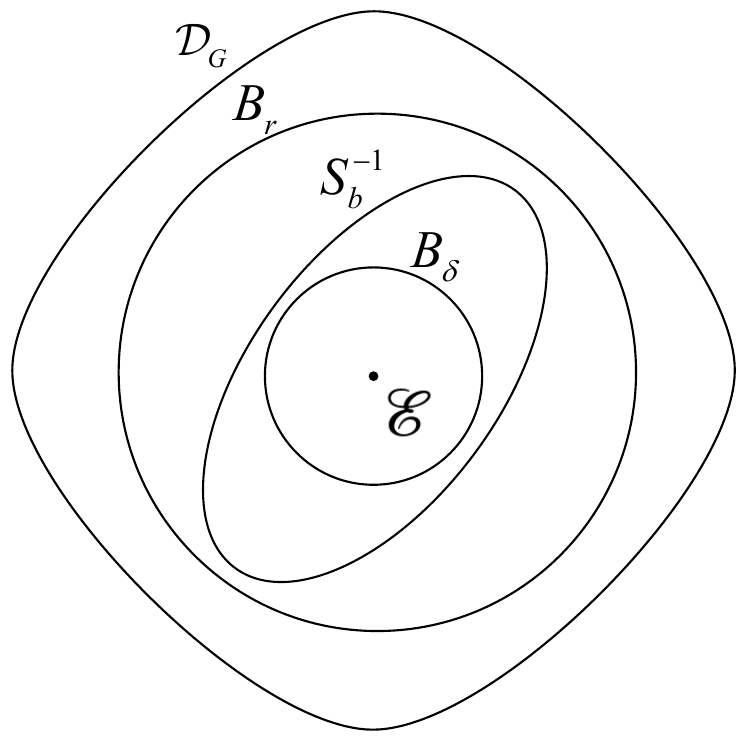}
			\caption{Relations of the sets in the proof of Theorem \ref{th:c5-stability}.}
			\label{fig:C5-relation of sets}
		\end{figure}
		If not, $S_b^{-1}$ intersects the boundary $\partial B_r$ of $B_r$ at a point $(\hat{x},\hat{u})$. At that point $S(\hat{x},\hat{u})\geq a>b$, contradicting $(\hat{x},\hat{u})\in S_b^{-1}$.
		
		By the continuity of $S(x,u)$ and the fact that $S(x,u)=0,\forall (x,u)\in\mathscr{E}$, there exists $\delta>0$ such that
		\begin{equation*}
			\text{dist}\left((x,u),\mathscr{E}\right)\leq\delta\implies S(x,u)<b.
		\end{equation*}
		Hence, $$\mathscr{E}\subset B_\delta \subset S_b^{-1}\subset \text{Int}(B_r) \subset \mathcal{D}_G.$$
        It follows from \cite[Theorem 6]{10239448} that $S_b^{-1}$ is a positive invariant set and is an estimate of the $f$-RoA, i.e., $(x(0),u(0))\in S_b^{-1}$ implies $(x(t),u(t))\in S_b^{-1},\forall t\geq0$ and $f(x(t),u(t))\to0$ as $t\to\infty$. Thus, $\text{dist}((x(0),u(0)),\mathscr{E})<\delta$ implies $\text{dist}\left((x(t),u(t)),\mathscr{E}\right)<r\leq\varepsilon,\forall t\geq0$.
		
		By Assumption \ref{as:c5-eqexist}, $\mathscr{E}$ is bounded. Therefore, $f(x(t),u(t))\to0$ implies $\text{dist}((x(t),u(t)),\mathscr{E})\to0$. Thus, $\text{dist}((x(0),u(0)),\mathscr{E})<\delta$ implies $\text{dist}\left((x(t),u(t)),\mathscr{E}\right)\to0$ as $t\to\infty$. This proves that $\mathscr{E}$ is asymptotically stable.
		
		2) Suppose $(x^*,u^*)\in\mathscr{E}$ is an isolated equilibrium point. Since $(x^*,u^*)$ is isolated, there exists sufficiently small $b>0$ such that a connected branch of the sub-level set $S_b^{-1}$ contains no other equilibrium point except $(x^*,u^*)$. Therefore, by the same argument of 1), one can prove that $(x^*,u^*)$ is asymptotically stable.
		
		3) It suffices to show that ${S}_l^{-1}\subset\mathcal{D}_G$ for any level value $l<\bar{l}$. Since $\mathscr{E}$ is nonempty, the intersection of any connected branch of ${S}_l^{-1}$ with $\mathcal{D}_G$ is nonempty. Thus, if ${S}_l^{-1}\not\subset\mathcal{D}_G$, then there must be a point $(\hat{x},\hat{u})$ where ${S}_l^{-1}$ and $\partial\mathcal{D}_G$ intersect. At that point $S(\hat{x},\hat{u})\geq \bar{l}>l$, which contradicts $(\hat{x},\hat{u})\in S_l^{-1}$. Therefore, for any $l<\bar{l}$, ${S}_l^{-1}\subset\mathcal{D}_G$. Thus by \cite[Theorem 6]{10239448}, $S_l^{-1}$ is a positive invariant set and an estimate of $f$-RoA. Thus, if $(x(0),u(0))\in{S}_l^{-1}$, then $(x(t),u(t))\in{S}_l^{-1}\subset\mathcal{D}_G$, $\forall t\geq0$, and $f(x(t),u(t))\to0$ as $t\to\infty$. Since $\mathscr{E}$ is bounded by Assumption \ref{as:c5-eqexist}, $f(x(t),u(t))\to0$ implies $(x(t),u(t))\to\mathscr{E}$, which completes the proof.
	\end{proof}
	\begingroup
	 
	\begin{remark}[\textit{Periodic behavior}]
		Within a certified positively invariant sublevel set, $\dot{S}\leq-\gamma(\|f\|)$ excludes any nontrivial periodic trajectory: integration over one period would contradict the periodicity of $S$ unless $f\equiv0$. Hence, all certified trajectories converge to $\mathscr{E}$, and a local Hopf bifurcation cannot occur while the strict certificate remains valid in the relevant neighborhood. Locating or classifying bifurcation boundaries is beyond the scope of this theorem; see \cite{khalil2002nonlinear} for the underlying Lyapunov argument.
	\end{remark}
	\endgroup
	The above theorem links LDD to the stability of an equilibrium set. It decomposes the system-wide stability into local LDD conditions and the coupling condition \eqref{eq:c5-global}. Each subsystem can verify the local condition without knowing the equilibrium that depends on the interconnection. The coupling condition \eqref{eq:c5-global} abstracts each subsystem by the dissipative matrix $X_i$ and involves the network differential information $\partial h_{net}/\partial y$. 
	
	We remark that the coupling condition \eqref{eq:c5-global} should hold for all $y\in\mathcal{D}_y$, which might be difficult to verify for a general nonlinear coupling function $h_{net}$. Fortunately, the coupling equations of the power network \eqref{eq:c5-ps-HN} are linear, since we choose the port variables $u$ and $y$ as voltages and currents. In such a case, $\partial h_{net}/\partial y$ is a constant matrix and hence \eqref{eq:c5-global} is independent of $y$.  

	\section{{Compositional} and Equilibrium-free Stability Conditions for Power Systems}

	In this section, we present {compositional} and equilibrium-free conditions to certify the stability of power system \eqref{eq:c5-system}.
	\subsection{Stability Conditions}\label{sec:stability conditions for power system}
	Consider the following conditions for power system \eqref{eq:c5-system}.
	\begin{condition}[\textit{Dynamic bus}]\label{con:c5-3}
		For each {$j\in\mathcal{V}_1$}, the dynamic subsystem \eqref{eq:c5-ps-dbus} satisfies LDD($X_j,\mathcal{D}_j$) for some symmetric matrix $X_j\in\rr^{4\times4}$ and some region $\mathcal{D}_j\subset\rr^{n_j}\times\rr^{2}$, with the corresponding storage function $S_j(x_j,u_j)$.
	\end{condition}
	\begin{condition}[\textit{Static bus}]\label{con:c5-4}
		For each {$k\in\mathcal{V}_2$}, the static subsystem \eqref{eq:c5-ps-sbus} satisfies LDD($X_k,\mathcal{D}_k$) for some symmetric matrix $X_k\in\rr^{4\times4}$ and some region $\mathcal{D}_k\subset\rr^{2}$.
	\end{condition}
	\begin{condition}[\textit{Coupling}]\label{con:c5-5}
		There exist $p_i>0$, $i\in\mathcal{V}$ such that
		\begin{equation}\label{eq:couple-C}
			\begin{bmatrix}
				-C\\I
			\end{bmatrix}^\text{T}P_\pi^\text{T}\text{blkdiag}(p_1X_1,\dots,p_{N}X_{N})P_\pi\begin{bmatrix}
				-C\\I
			\end{bmatrix}\preceq0,
		\end{equation}
		where $C$ is a constant matrix defined as in \eqref{eq:c5-partialhnet}, and $P_\pi$ is the permutation matrix such that $P_\pi\col(u,y)=\col(u_1,y_1,\dots,u_{N},y_{N})$.
	\end{condition}
	\begin{remark}
		The algebraic coupling condition requires only the network parameters and the dissipativity metrics $X_i$ of each grid-connected device. Since no detailed device model is required, it provides a privacy-preserving framework for incorporating massive, heterogeneous devices into large power systems.
		{Because the coupling condition represents the complete network interconnection, its dimension necessarily grows with the number of subsystems. The compositional formulation mitigates this growth by confining each device's internal dynamics and storage-function variables to its local LDD problem and representing the device at the network level only through the fixed-dimensional matrix $X_i\in\rr^{4\times4}$ associated with its two-axis voltage--current port. Hence, the model-dependent local problems are separable and can be solved independently or in parallel, while only certificate-level matrices are composed through the network condition. Increasing the order or complexity of an individual device therefore enlarges its local problem without propagating its full internal dimension into the network-level certificate. Part II, Sections III and IV-C, develop the corresponding computational procedure and demonstrate its practical feasibility on a modified IEEE 118-bus system.} {Part II, Section III-A, formulates the joint construction of $\{X_i\}_{i\in\mathcal V}$ as a feasibility problem combining the local LDD and network coupling conditions.} The coupling condition is also independent of equilibrium; it needs to be revalidated only when the network topology or the devices' dissipativity changes.
	\end{remark}
	\begin{theorem}[\textit{Stability of power systems}]
    \label{th:c5-xiet}
		Consider the power system \eqref{eq:c5-system} satisfying {Assumptions \ref{as:c5-network-wellposed}-\ref{as:c5-nonsingular}}. If Conditions \ref{con:c5-3}-\ref{con:c5-5} are satisfied, then the following assertions hold:
		\begin{enumerate}
			\item The equilibrium set $\mathscr{E}$ is asymptotically stable;
			\item Every isolated equilibrium point in $\mathscr{E}$ is asymptotically stable.
		\end{enumerate}
	\end{theorem}
	\begin{proof}
		It is a direct corollary of Theorem \ref{th:c5-stability}.
	\end{proof}
	{Theorem \ref{th:c5-xiet} provides a compositional, equilibrium-independent approach to analyzing power-system stability. Each device can first verify locally that it satisfies LDD($X_i,\mathcal{D}_i$). The resulting $X_i$ certificates are then combined through the network-level coupling condition, whose verification establishes their global consistency. If all three conditions are satisfied, one can directly conclude that the equilibrium set $\mathscr{E}$ of the entire power system is asymptotically stable and that any isolated equilibrium in $\mathscr{E}$ is asymptotically stable. In this specific sense, the certificate is not tied to a linearization about one preselected equilibrium: the same regional certificate applies to all equilibria contained in the certified region.}
	\subsection{Discussion on Equilibrium-Free Property}
	The salient equilibrium-free feature of our approach is reflected in two aspects: requirement and result. 
	
	{First, the local stability requirements are free of a preselected system-wide equilibrium. Neither local Condition 1 nor Condition 2 requires the device to know the equilibrium imposed by the complete interconnection. This facilitates constructing device certificates over prescribed state--input regions before a particular system-wide operating equilibrium is selected.} 

	{Second, the stability conclusion is not centered at one specified equilibrium. The same feasible regional certificate establishes stability for the equilibrium set contained in $\mathcal{D}$, which is useful when the interconnected system admits multiple equilibria or an equilibrium continuum.}
	
	{The proposed method is intended to complement, rather than replace, equilibrium-specific small-signal analysis. Eigenvalue analysis is direct and effective when a particular equilibrium and the complete linearized model are available. By contrast, the proposed method constructs a nonlinear regional sufficient certificate from local device properties and a network condition. The same certificate applies to all system-consistent equilibria contained in the certified region and therefore provides a reusable stability envelope when multiple equilibria coexist or the operating equilibrium varies with load and generation conditions, through a generally more demanding common regional construction that addresses a different certification task. Unlike classical direct methods based on a system-wide energy function, the aggregate storage function here is assembled from local storage functions; relative to global passivity or dissipativity conditions, LDD trades global validity for applicability over prescribed operating regions.}

	{This distinction is particularly important in compositional analysis. The local equilibrium of an individual device is determined by the complete network interconnection and is therefore not purely local information. Constructing each device certificate around one operating equilibrium would require the corresponding globally determined equilibrium information and would limit reuse when the operating condition changes. In the proposed framework, the device-level LDD certificates are instead constructed over prescribed regions without preselecting a system-wide equilibrium and are subsequently composed through the network condition. This equilibrium independence applies to certificate construction: to certify one particular operating equilibrium, it must still be a system-consistent equilibrium contained in the certified region.}

	{Beyond the isolated equilibria considered in Section \ref{sec:case}, the equilibrium-set conclusion is particularly relevant when the equilibria form a nonisolated continuum. A smooth equilibrium continuum necessarily introduces a tangent zero mode, so a standard strict-Hurwitz test for an isolated equilibrium is not directly applicable, and an individual point on a nontrivial continuum cannot be asymptotically stable. A power-system example in \cite{10239448} exhibits a continuous equilibrium trajectory to which solutions starting from different initial conditions converge at different points. This illustrates why asymptotic stability of the equilibrium set in Theorem \ref{th:c5-xiet} is distinct from its additional pointwise conclusion for an isolated equilibrium.}

	\subsection{{Compositional Analysis} for Transient Stability}\label{sec:transient stability}
	When the proposed conditions are satisfied, the following theorem shows that sublevel sets of the storage function can be used to estimate the region of attraction (ROA) to the equilibrium set.
	\begin{theorem}[\textit{Estimation of RoA}]
    \label{th:c5-region}
		Consider a power system \eqref{eq:c5-system} that satisfies {Assumptions \ref{as:c5-network-wellposed}-\ref{as:c5-nonsingular}} and Conditions \ref{con:c5-3}-\ref{con:c5-5}. Let {$S(x,u)=\sum_{j\in\mathcal{V}_1}p_jS_j(x_j,u_j)$} and $\bar{l}:=\min_{(x,u)\in\partial\mathcal{D}_G}S(x,u)$. Then for all level values $l<\bar{l}$, any trajectory $(x(t),u(t))$ of \eqref{eq:c5-system} starting from ${S}_l^{-1}$ has the following properties:
		\begin{enumerate}
			\item $(x(t),u(t))\in\mathcal{D}_G$, $\forall t\geq0$;
			\item $(x(t),u(t))\to\mathscr{E}$ as $t\to\infty$.	
		\end{enumerate}
	\end{theorem}
	\begin{proof}
		It is a direct corollary of Theorem \ref{th:c5-stability}. 
	\end{proof}
	The above theorem provides a direct method to assess the transient stability of the power system without knowing the post-fault equilibrium. Suppose the power system \eqref{eq:c5-system} is subjected to a large disturbance that causes the state to deviate from the original stable operating point. Let the post-disturbance initial condition be $(x^0,u^0)$. If $(x^0,u^0)\in\mathcal{D}_G$ and $S(x^0,u^0)<\bar{l}$, then the post-disturbance trajectory remains in $\mathcal{D}_G$ and eventually converges to the equilibrium set, i.e., the system restores stable operation. Note that we do not require knowledge of the post-fault equilibrium in the assessment. Consequently, the conclusion of the assessment is convergence towards the post-fault equilibria set rather than towards any specific equilibrium. Such equilibrium independence is advantageous when the post-fault equilibrium is difficult to determine ex ante, especially when multiple equilibria exist. We refer interested readers to \cite{10239448} for a more detailed discussion on equilibrium-independent transient stability analysis of power systems.
	
	{Theorem \ref{th:c5-region} also admits a compositional interpretation. The function $S(x,u)$ and the region $\mathcal{D}$ are assembled from local objects: $S(x,u)$ is the weighted sum of the storage functions $S_i(x_i,u_i)$ of the dynamic subsystems, and $\mathcal{D}$ is the Cartesian product of the regions $\mathcal{D}_i$ of the dynamic and static subsystems. Hence, the network coupling need not be considered when constructing these local objects. The critical level $\bar{l}=\min_{(x,u)\in\partial\mathcal{D}_G}S(x,u)$ is then evaluated after composing the local objects with the algebraic manifold $\mathscr{G}$ of the complete system.}

	\begingroup
	 
	\begin{remark}[\textit{Conservativeness and trade-off}]
		Because the proposed conditions are sufficient, failure to obtain a feasible certificate does not disprove stability of the equilibrium set. Here, conservativeness is understood relative to the equilibrium-set certification task: the compositional decomposition and the requirement of one common regional certificate may limit the equilibria and operating conditions covered, while the sublevel set in Theorem \ref{th:c5-region} is an inner estimate of the region of attraction of the equilibrium set. A direct comparison with equilibrium-wise eigenvalue analysis is not like-for-like, because the latter assesses local stability at one preselected equilibrium. In return, the proposed certificate can be constructed without a preselected system-wide equilibrium, reused across compatible equilibria and operating conditions, and applied to finite-amplitude post-disturbance states within the certified sublevel set.
	\end{remark}
	\endgroup
	
	Besides stability assessment, Theorem \ref{th:c5-region} also points out a localized control principle to enhance transient stability. Theorem \ref{th:c5-region} indicates that for the same $S(x,u)$, a larger region $\mathcal{D}_G$ guarantees a larger region of attraction, and hence achieves better transient stability. Since the algebraic constraints are fixed, expanding $\mathcal{D}_G$ essentially requires expanding $\mathcal{D}_i$ of each subsystem, i.e., enlarging the region in which each grid-connected device satisfies LDD. Therefore, a localized control target to enhance the system-wide transient stability would be to maximize $\mathcal{D}_i$ without changing $S_i(x_i^0,u_i^0)$ and $X_i$.
	
	\section{An Illustrative Example}\label{sec:case}
	{To transparently illustrate how the proposed theory proceeds from device-level certificates to a system-wide stability conclusion, we consider a single-machine single-load (SMSL) power system, as shown in Fig. \ref{fig:c5-2bus}. This deliberately small benchmark is an inspectable theorem-level example: it allows the assumptions, the dynamic- and static-device LDD certificates, the network coupling inequality, the invariant sublevel set, and the equilibrium-set stability conclusion to be examined explicitly. Broader validation is provided in Part II, Section IV, where the same theory underpins applications to modified IEEE 9-bus, 39-bus, and 118-bus systems.}
	\begin{figure}[htb]
		\centering
		\includegraphics[width=.8\hsize]{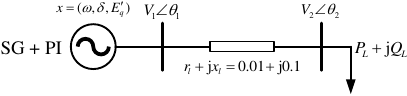}
		\caption{The SMSL 2-bus system. An SG with a frequency PI regulator is connected to a load through a transmission line with impedance $r_l + \text{j}x_l$.}
		\label{fig:c5-2bus}
	\end{figure}
	\subsection{Settings}
	The system consists of a constant PQ load and an SG modeled by 3rd-order flux-decay dynamics, as given in Example \ref{eg:c5-sg}.
	Table \ref{tab:c5-2bus-sg} reports the parameters of the SG. The SG can be represented in the standard input-output form with $u_1=\col(V_{D1},V_{Q1})$ and $y_1=\col(-I_{D1},-I_{Q1})$, while the constant-PQ load can be represented in the standard input-output form with $u_2=\col(-I_{D2},-I_{Q2})$ and $y_2=\col(V_{D2},V_{Q2})$. Details of this model and the input-output transformation are provided in Part II. {For reporting the numerical values, the angular origin of the common $DQ$ frame is selected such that the nominal bus-1 voltage lies on the $D$ axis, i.e., $(V_{D1}^0,V_{Q1}^0)=(1,0)$. This fixes only the coordinate convention; it is not a requirement that bus 1 have zero angle at every equilibrium.}
	{Accordingly, we focus here on illustrating the complete certificate-to-conclusion chain and on developing deeper insights into equilibrium independence.}
	\begin{table}[htb]
		\centering
		\footnotesize
		\caption{Parameters of SG in the 2-bus system (values in p.u.).}
			\begin{tabular}{ccccccccc}
				\toprule
				$M$& $D$& $T'_{d0}$& $x_{d}$& $x_{q}$& $x'_{d}$& $P^m$& $E_f$ & $K_I$\\
				\midrule
				0.41& 0.3& 5.4& 0.67& 0.40& 0.13& 0.48& 1.11 & 0.5\\ 
				\bottomrule
			\end{tabular}
			\label{tab:c5-2bus-sg}
	\end{table}
	\subsection{Verifying Equilibrium-Free Stability Conditions}
	Now we illustrate how to use the proposed {compositional} conditions to verify system-wide stability when local devices do \emph{not} know the system-wide equilibrium. First, for the dynamic subsystem, i.e., the SG at bus 1, we verify Condition \ref{con:c5-3}. {Specifically, we select a Krasovskii-type storage function $S_1=f_1^\T\mathcal{P}_1f_1$ with $\mathcal{P}_1\succ0$, translate the LDD inequalities into matrix inequalities over the prescribed state--input region, and solve the resulting feasibility problem for $\mathcal{P}_1$ and the supply-rate matrix $X_1$. The detailed derivation and computational procedure are provided in Part II, Section II-B and Theorem 1.}
	We find that the SG is LDD($X_1,\mathcal{D}_1$) with the following $\mathcal{P}_1$ and $X_1$.
	\begin{equation*}
		\mathcal{P}_1=\begin{bmatrix}
			92.44&	28.52&	-36.55\\
			28.52&	82.60&	54.99\\
			-36.55&	54.99&	1000
		\end{bmatrix},
	\end{equation*}
	\begin{equation*}
		{ X_1=\begin{bmatrix}
			103.6171&	-97.6429&	-510.8229&	-92.6929\\
			-97.6429&	-33.4829&	98.2271&	-4.9229\\
			-510.8229&	98.2271&	999.9971&	-7.1129\\
			-92.6929&	-4.9229&	-7.1129&	6.8871
		\end{bmatrix}.}
	\end{equation*}

	The dissipative region $\mathcal{D}_1$ is a nonlinear region in five-dimensional space; its projection is shown in Fig.~\ref{fig:2bus_3D}.
	
	For the static subsystem constant PQ load at bus 2, we need to verify Condition \ref{con:c5-4}. Solving \eqref{eq:c5-Ps}, we find that it is LDD with
	\begin{equation*}
		{ X_2=\begin{bmatrix}
			-990.6535&	60.0465&	-503.6435&	100.9865\\
			60.0465&	-5.8935&	-73.2835&	-6.2335\\
			-503.6435&	-73.2835&	-107.6435&	93.5965\\
			100.9865&	-6.2335&	93.5965&	29.4065
		\end{bmatrix}.}
	\end{equation*}
	
	Now we check the coupling Condition \ref{con:c5-5}. Based on the input-output relation of the network, we obtain
	\begin{equation*}
		C=\begin{bmatrix}
			r_l &  -x_l  & -1   &      0\\
			x_l &  r_l  &  0 &  -1\\
			1   &      0   &      0    &     0\\
			0   & 1   &     0     &    0
		\end{bmatrix}.
	\end{equation*}
	Substituting $X_1$ and $X_2$ into \eqref{eq:couple-C}, we find the coupling condition is justified with $p_1=p_2=1$.
	
	We emphasize that our method is independent of equilibrium-point information during local condition verification. Conventional approaches---such as those requiring model linearization at a subsystem's equilibrium---depend on precise knowledge of the local equilibrium $(x_i^*,u_i^*)$, which is inherently coupled to all subsystems. Consequently, our method does not aim to assert stability of a specific equilibrium. Instead, it directly asserts stability of all equilibria in the dissipative region, as we show below.
	\subsection{{Compositional Stability Certification}}
	It follows from Theorem \ref{th:c5-xiet} that any possible isolated equilibrium $(x^*,u^*)$ in $\mathcal{D}_1\times\mathcal{D}_2$ is asymptotically stable. Compositionally, we only need to check $(x^*,u^*_1)\in\mathcal{D}_1$ and $u^*_2\in\mathcal{D}_2$. To illustrate, we first consider the nominal load case with $P_L^0=0.4$ and $Q_L^0=0.1$. Given these parameters, the 2-bus nonlinear power system has two isolated equilibria. The first equilibrium reads $$x^*=(\delta^*,\omega^*,E_{q}^{'*})=(0.1527,0,1.0118),$$ 
    $$u_1^*=(V_{D1}^*,V_{Q1}^*)=(1,0),$$ 
    $$u_2^*=-(I_{D2}^*,I_{Q2}^*)=(0.4018,-0.1175).$$ 
    The second equilibrium reads $$x^*=(\delta^*,\omega^*,E_{q}^{'*})=(0.1231,0,0.4757),$$ 
    $$u_1^*=(V_{D1}^*,V_{Q1}^*)=(0.3443,-0.1701),$$ 
    $$u_2^*=-(I_{D2}^*,I_{Q2}^*)=(0.6664,-1.1001).$$
	
    Table \ref{tab:c5-2bus-equilibria}, Column 2, reports the certification result of our method: the first equilibrium belongs to $\mathcal{D}_1\times\mathcal{D}_2$ and is therefore certified as stable, whereas the second (low-voltage) equilibrium does not belong to $\mathcal{D}_1\times\mathcal{D}_2$. This result is consistent with the centralized eigenvalue analysis, which shows that the first equilibrium is stable while the second is unstable. 
	\begin{table}[htb]
		\centering
		\footnotesize
		\caption{Stability certification of different equilibria in the 2-bus system.}
		\resizebox{\columnwidth}{!}{%
\begin{tabular}{cccc}
			\toprule
			\multirow{2}{*}{Equilibrium}  & \multicolumn{2}{c}{{Compositional}} &Centralized\\
			& $(x^*,u_1^*)\in\mathcal{D}_1?$ & $u_2^*\in\mathcal{D}_2?$ & $\max\text{Re}\lambda(x^*,u^*)<0?$\\
			\midrule
			1 & \checkmark & \checkmark & \checkmark, stable\\ 
			2 & $\times$ & $\times$ & $\times$, unstable\\ 
			\bottomrule
		\end{tabular}%
		}
			\label{tab:c5-2bus-equilibria}
	\end{table}
	
	\subsection{Transient Stability Analysis}
	The theoretical result in Section \ref{sec:transient stability} provides a direct method for assessing transient stability. The sublevel set of $S(x,u)$ within the dissipative region $\mathcal{D}_G$ provides an {inner estimate of the region of attraction of the equilibrium set}. Any trajectory starting within this sublevel set converges to the equilibrium set.
	
	To verify and illustrate, we estimate the stability region of the 2-bus system, as visualized in Fig.~\ref{fig:2bus_3D}. Here, Fig.~\ref{fig:2bus_3D}(A) depicts the local dissipative region $\mathcal{D}_G$. The largest sublevel set $S_l^{-1}$ within it (with $l=0.2104$) is the estimated stability region. Note that $\mathcal{D}_G$ is independent of $\omega$, since $\partial f/\partial x$ does not depend on $\omega$. The figure also shows the equilibrium point $x^*=(0.1527,0,1.0118)^\T$ and the trajectory starting from $x^0=(0.16,0.02,0.92)^\T\in S_l^{-1}$ in three-dimensional space. As predicted by our theory, this trajectory converges to $x^*$. Fig.~\ref{fig:2bus_3D}(B) shows a two-dimensional slice at $\omega=0$, providing a clearer view of the region boundaries. In addition, Fig.~\ref{fig:2bus_3D}(C) shows the time-domain trajectories of $\Delta x(t)=x(t)-x^*$, which also verifies our theory.
	
	\begin{figure*}[htb]
		\centering
		\includegraphics[width=1\hsize]{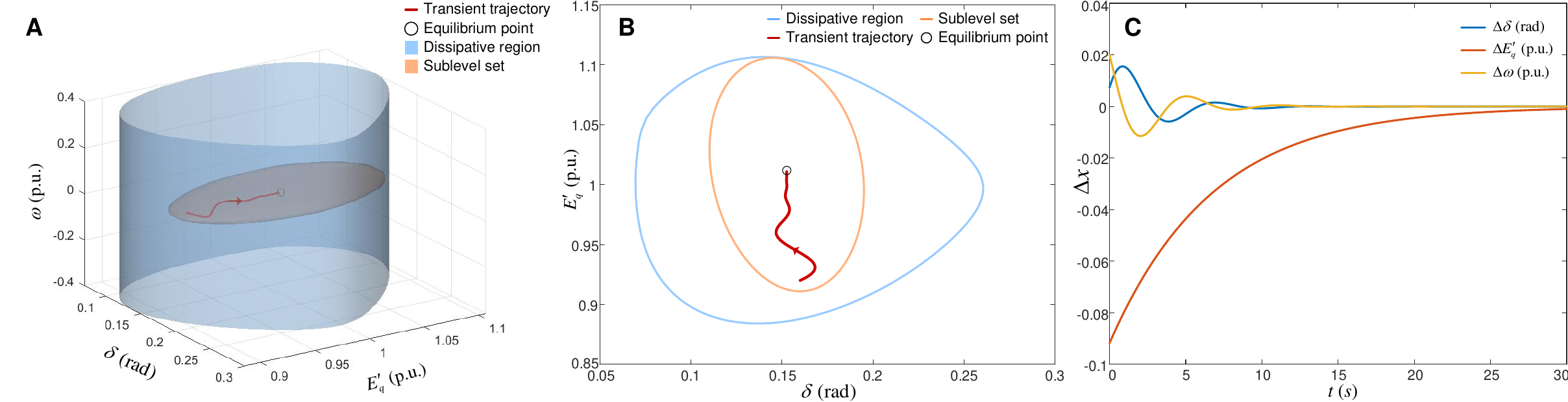}
		\caption{The stability region estimation of the 2-bus system. (A) The local dissipative region $\mathcal{D}_G$ and the sublevel set $S_l^{-1}$ (with $l=0.2104$) as the guaranteed stability region, along with the trajectory starting from $x^0=(0.16,0.02,0.92)^\T$ and the equilibrium point $x^*=(0.153,0,1.012)^\T$, shown in three-dimensional space. (B) The two-dimensional slice of subfigure~(A) with $\omega=0$. (C) Time-domain trajectories of $\Delta x(t)=x(t)-x^*$ starting from $x^0$.}
		\label{fig:2bus_3D}
	\end{figure*}
	
	\subsection{Stability Under Varying Load}
	Now we illustrate how the equilibrium-free property facilitates stability analysis under varying operating conditions. Consider that the load at Bus 2 scales by a factor $s>0$, i.e., $P_L+\text{j}Q_L=s(P_L^0+\text{j}Q_L^0)$. This load variation shifts the system's equilibrium. With the same $\mathcal{P}_i$ and $X_i$, the {compositional stability certificates} and the local dissipative region $\mathcal{D}_1$ remain unchanged. To ensure the stability of different equilibria under different operating conditions, we check whether each new equilibrium point remains within the local dissipative region. This avoids reconstructing an equilibrium-specific stability certificate, such as repeating eigenvalue analysis, for every new equilibrium point. 
	
	We trace the continuous changes in equilibrium as the load factor $s$ varies. The orange curve in Fig. \ref{fig:2bus-vary-A} shows this equilibrium trajectory. The blue curve shows the projection of dissipative equilibrium region $\mathcal{D}_1^{eq}=\{(x,u_1)\in\mathcal{D}_1|f(x,u_1)=0\}$ on the $u_1$ plane. We found that when $s$ ranges between 0.902 and 1.075, the corresponding equilibrium points remain within $\mathcal{D}_1^{eq}$, as shown in Fig. \ref{fig:2bus-vary-A}. Calculation also shows that the equilibria are inside $\mathcal{D}_2$ for $s\in(0.787,1.153)$. Hence, our theory ensures the system-wide stability for all $s\in(0.902,1.075)$ without reconstructing an equilibrium-specific stability certificate for each value of $s$. Fig. \ref{fig:2bus-vary-B} shows the time-domain trajectories when the load step changes, which also verifies our analysis. 
	
	\begin{figure}[ht]
		\centering
		\includegraphics[width=0.85\hsize]{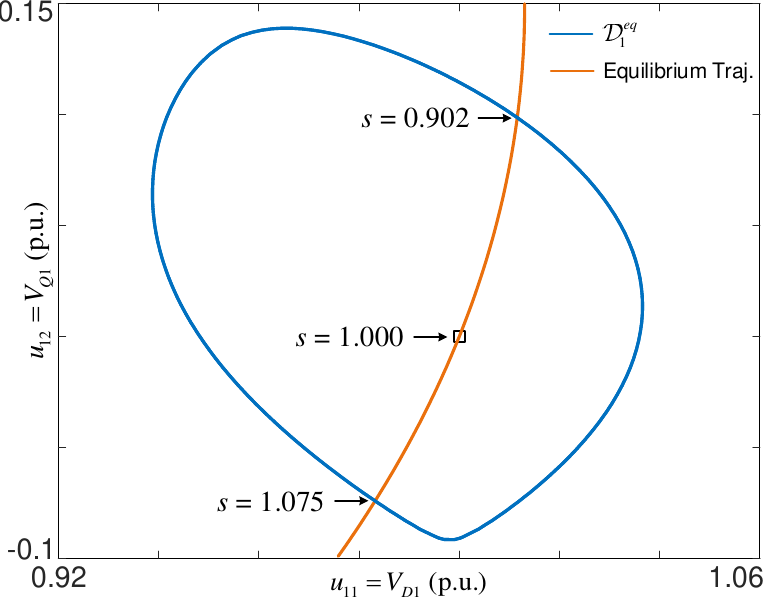}
		\caption{Local delta dissipative region $\mathcal{D}_1^{eq}$ in the $u_1$-plane and the equilibrium trajectory as the load factor $s$ varies. For $s\in[0.902,1.075]$, the corresponding equilibria are inside the region and hence are stable.}
		\label{fig:2bus-vary-A}
	\end{figure}
	
	
	\begin{figure}[ht]
		\centering
		\includegraphics[width=0.9\hsize]{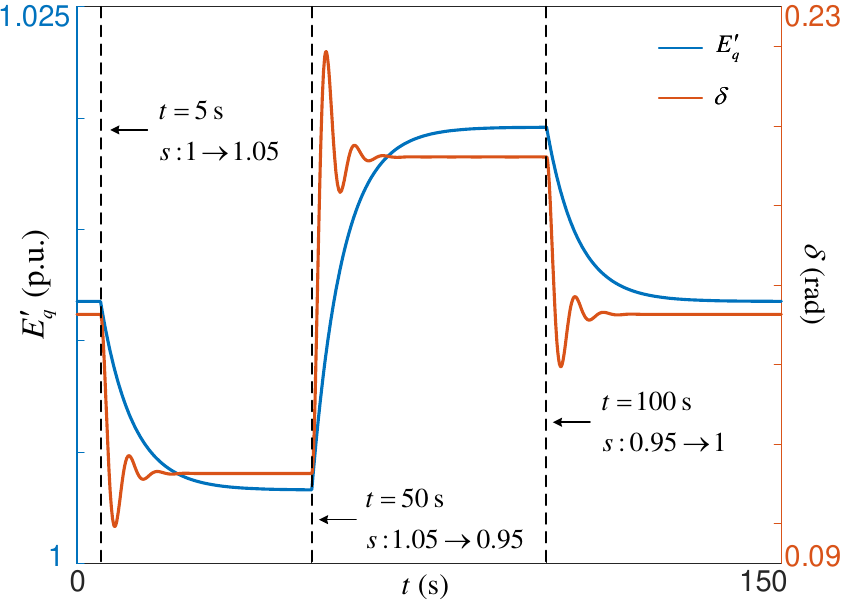}
		\caption{The load step changes from the nominal value $s=1$ to $s=1.05$, then to $s=0.95$, and finally to $s=1$, at time $t=5$s, $t=50$s, and $t=100$s, respectively. The system remains stable under different operating conditions.}
		\label{fig:2bus-vary-B}
	\end{figure}

	\begingroup
	 
	\subsection{Certification Under Model-Parameter Uncertainty}
	We further examine the robustness of our method against variations in the SG dynamic parameters. Specifically, we let
	\begin{equation}
		M\in[0.9M_0,1.1M_0],\qquad
		T'_{d0}\in[0.9T'_{d0,0},1.1T'_{d0,0}],
	\end{equation}
	where $M_0=0.4138$ and $T'_{d0,0}=5.4\,\mathrm{s}$, and keep the same $\mathcal P_1$, $X_1$, and $X_2$ reported above. No uncertainty-aware certificate is constructed. These two parameters alter the dynamic model but not the equilibrium equations or the network interconnection.

	The robust guarantee over the continuous parameter box can be verified because the SG local LDD condition is an affine matrix inequality in $\alpha=1/M$ and $\beta=1/T'_{d0}$. Every point in the corresponding \((\alpha,\beta)\) rectangle is a convex combination of its four vertices. Therefore, satisfaction of the LDD condition at the four vertices guarantees its satisfaction throughout the continuous parameter box.
	
	The calculation shows that the same matrices \(\mathcal P_1\), \(X_1\), and \(X_2\) certify all parameter values in the above continuous box at the nominal operating condition \(s=1\). Together with the unchanged PQ-load and network coupling conditions, this verifies the same compositional stability certificate for every \(M\) and \(T'_{d0}\) in the \(\pm10\%\) ranges.
	
	We further repeat the four-vertex verification along the varying-load equilibrium branch. Requiring the same certificate to hold uniformly over the entire prescribed parameter box, we find that the jointly certified operating interval is \(s\in[0.953,1.024]\). Thus, the fixed nominal certificate provides a continuous guarantee with respect to the prescribed parameter uncertainty and remains reusable over a nontrivial range of operating conditions. Recall that, in the previous subsection, our method certifies stability over \(s\in[0.902,1.075]\) for the nominal parameter values. This comparison quantifies the reduction in the certified operating range caused by the parameter uncertainty.
	\endgroup
	
	{In summary, the SMSL example demonstrates the compositional structure of the proposed system-wide stability analysis. Each bus certificate can be verified without knowing the system-wide equilibrium, while the certificates are combined through the network-level coupling condition. The resulting stability certification is equilibrium-set-oriented; that is, all equilibria in the dissipative region are guaranteed to be stable.}
	
	\section{Concluding Remarks}
	{In Part I of this paper, we developed a compositional, equilibrium-free stability analysis for nonlinear power systems. We proved that system-wide stability with respect to an equilibrium set can be decomposed into (i) device-level LDD conditions, which do not require explicit knowledge of the system-wide equilibrium point, and (ii) a network-level coupling condition that establishes global consistency among the device certificates.}
	{The proposed stability conditions are compatible with structure-preserving models of lossy power systems and enable compositional analysis without network reduction or restrictive assumptions on device properties. The single-machine single-load benchmark illustrates this compositional and equilibrium-set-oriented stability analysis.}
	
	{The results in Part I provide a theoretical foundation for more scalable and adaptable compositional stability analysis in modern power grids. Although the dimension of the network coupling condition grows with the number of subsystems, the proposed formulation mitigates the effect of high-dimensional device models by retaining their dynamics and storage-function variables in separable local LDD problems and composing only fixed-dimensional certificate matrices at the network level. Part II develops the corresponding certificate-construction procedure and demonstrates its practical feasibility on modified IEEE benchmarks.}
	{A systematic treatment of general parameter uncertainty and unmodeled dynamics remains an important direction for future work.}
	%
	%
	\bibliographystyle{IEEEtran}
	\bibliography{mybib}
	
	
	%
	
	
\end{document}